\documentclass[11pt,reqno]{amsart}

\usepackage{color}
\usepackage[colorlinks=true, allcolors=blue]{hyperref}
\usepackage{amsmath, amssymb, amsthm}
\usepackage{mathrsfs}
\usepackage{mathtools}
\usepackage[noabbrev,capitalize,nameinlink]{cleveref}
\crefname{equation}{}{}
\usepackage{fullpage}
\usepackage[noadjust]{cite}
\usepackage{graphics}
\usepackage{pifont}
\usepackage{tikz}
\usepackage{bbm}
\usepackage[T1]{fontenc}

\DeclareMathOperator*{\argmin}{arg\,min}
\usetikzlibrary{arrows.meta}

\usepackage{environ}
\usepackage{framed}
\usepackage{url}
\usepackage[linesnumbered,ruled,vlined]{algorithm2e}
\usepackage[noend]{algpseudocode}
\usepackage[labelfont=bf]{caption}
\usepackage{cite}
\usepackage{framed}
\usepackage[framemethod=tikz]{mdframed}
\usepackage{appendix}
\usepackage{graphicx}
\usepackage[textsize=tiny]{todonotes}
\usepackage{tcolorbox}
\allowdisplaybreaks[1]

\crefname{algocf}{Algorithm}{Algorithms}

\crefname{equation}{}{} 
\AtBeginEnvironment{appendices}{\crefalias{section}{appendix}} 

\usepackage{enumerate}
\usepackage[color, final]{showkeys} 

\colorlet{refkey}{orange!20}
\colorlet{labelkey}{blue!30}

\crefname{algocf}{Algorithm}{Algorithms}

\numberwithin{equation}{section}
\newtheorem{theorem}{Theorem}[section]
\newtheorem{proposition}[theorem]{Proposition}
\newtheorem{lemma}[theorem]{Lemma}
\newtheorem{claim}[theorem]{Claim}
\crefname{claim}{Claim}{Claims}

\newtheorem{corollary}[theorem]{Corollary}

\newtheorem*{question*}{Question}

\theoremstyle{definition}
\newtheorem{definition}[theorem]{Definition}
\newtheorem{problem}[theorem]{Problem}

\newtheorem*{definition*}{Definition}

\theoremstyle{remark}
\newtheorem*{remark}{Remark}

\newcommand{\mb}{\mathbb}
\newcommand{\mbf}{\mathbf}

\newcommand{\mc}{\mathcal}
\newcommand{\mf}{\mathfrak}

\newcommand{\ol}{\overline}
\newcommand{\on}{\operatorname}

\newcommand{\h}[1]{\widehat{#1}}

\allowdisplaybreaks
\title{Towards the sampling Lov\'asz Local Lemma}

\author[A1]{Vishesh Jain}
\address{Simons Institute for the Theory of Computing, Berkeley, CA 94720, USA}
\email{visheshj@stanford.edu}

\author[A2]{Huy Tuan Pham}
\author[A3]{Thuy Duong Vuong}
\address{Stanford University, Stanford, CA 94305, USA}
\email{\{huypham, tdvuong\}@stanford.edu}

\begin{document}

\begin{abstract}
Let $\Phi = (V, \mc{C})$ be a constraint satisfaction problem on variables $v_1,\dots, v_n$ such that each constraint depends on at most $k$ variables and such that each variable assumes values in an alphabet of size at most $[q]$. Suppose that each constraint shares variables with at most $\Delta$ constraints and that each constraint is violated with probability at most $p$ (under the product measure on its variables). We show that for $k, q = O(1)$, there is a deterministic, polynomial time algorithm to approximately count the number of satisfying assignments and a randomized, polynomial time algorithm to sample from approximately the uniform distribution on satisfying assignments, provided that 
\[C\cdot q^{3}\cdot k \cdot p \cdot \Delta^{7} < 1, \quad \text{where }C \text{ is an absolute constant.}\]
Previously, a result of this form was known essentially only in the special case when each constraint is violated by exactly one assignment to its variables.  

For the special case of $k$-CNF formulas, the term $\Delta^{7}$ improves the previously best known $\Delta^{60}$ for deterministic algorithms [Moitra, J.ACM, 2019] and $\Delta^{13}$ for randomized algorithms [Feng et al., arXiv, 2020]. For the special case of properly $q$-coloring $k$-uniform hypergraphs, the term $\Delta^{7}$ improves the previously best known $\Delta^{14}$ for deterministic algorithms [Guo et al., SICOMP, 2019] and $\Delta^{9}$ for randomized algorithms [Feng et al., arXiv, 2020].   
\end{abstract}

\maketitle

\section{Introduction}\label{sec:introduction}

The celebrated Lov\'asz Local Lemma (LLL) is a fundamental tool in probabilistic combinatorics which provides a sufficient condition for avoiding a collection of ``bad events'' in a probability space. In a quite general form, it may be stated as follows. 

\begin{theorem}[\cite{erdHos1973problems}]
\label{thm:local-lemma-general}
Let $\mc{C}$ be a finite set of events in a probability space. For $C \in \mc{C}$, let $\Gamma(C)$ denote a subset of $\mc{C}$ such that $C$ is independent of the collection of events $\mc{C}\setminus ({C} \cup \Gamma(C))$. Suppose there exist positive real numbers $x: \mc{C} \to (0,1)$ such that
\begin{align}
    \label{eqn:LLL-condition}
\mb{P}[C] \leq x(C)\prod_{D \in \Gamma(C)}(1-x(D))\quad \text{for all }C\in \mc{C}.
\end{align}
Then,
\[\mb{P}[\wedge_{C \in \mc{C}} \ol{C} ] \geq \prod_{C \in \mc{C}}(1-x(C)) > 0.\]
\end{theorem}

In most applications of the LLL (cf.~\cite{alon2004probabilistic, moser2010constructive, molloy1998further}), the underlying probability measure $\mb{P}[\cdot ]$ is generated by a collection of independent random variables $X_1,\dots, X_n$ and for each ``bad event'' $C \in \mc{C}$, there is a subset $\on{vbl}(C) \subseteq \{X_1,\dots, X_n\}$ such that $C$ depends only on $X_i\in \on{vbl}(C)$. This is often referred to as the ``variable-version'' setting of the LLL. Moreover, for many applications (cf.~\cite{alon2004probabilistic}), the following ``symmetric'' case of the variable-version setting suffices. 

\begin{corollary}
\label{cor:local-lemma-symmetric}
Let $X_1,\dots,X_n$ denote a collection of independent random variables. Let $\mc{C} = \{C_1,\dots, C_m\}$ denote a collection of events and for $C \in \mc{C}$, let $\on{vbl}(C)$ denote a subset of $\{X_1,\dots, X_n\}$ such that $C$ depends only on $X_i \in \on{vbl}(C)$. 
Suppose there exist $p \in (0,1)$ and $D\geq 0$ satisfying
\begin{itemize}
    \item For each $C \in \mc{C}$, $\mb{P}[C] \leq p$.
    \item For each $i \in [m]$, $\#\{j\in [m]: \on{vbl}(C_j)\cap \on{vbl}(C_i)\neq \emptyset\} \leq (D+1)$, and 
    \item $e\cdot p\cdot (D+1) \leq 1$, where $e$ is the base of the natural logarithm. 
\end{itemize}
Then, 
\[\mb{P}[\wedge _{i\in [m]}\ol{C_i}] \geq \prod_{i=1}^{m}\left(1-e\cdot \mb{P}[C_i]\right) > 0.\]
\end{corollary}

As a classical application of \cref{cor:local-lemma-symmetric}, consider the problem of satisfiability of a $k$-CNF formula over Boolean variables $x_1,\dots, x_n$.  Recall that a $k$-CNF formula over Boolean variables $x_1,\dots, x_n$ is a collection of constraints $C_1,\dots, C_m$ such that each $C_i$ depends on exactly $k$ variables and such that each $C_i$ is satisfied by all but exactly one assignment to its variables. \cref{cor:local-lemma-symmetric} shows that if each constraint shares variables with at most (roughly) $2^{k}/e$ other constraints, then the formula has a satisfying assignment.

Unfortunately, the original proof of \cref{thm:local-lemma-general} is non-constructive and does not provide an efficient algorithm to \emph{find} a satisfying assignment of the formula when this condition is met. In a breakthrough work \cite{beck1991algorithmic}, Beck showed that if $2^{k}/e$ is replaced by $2^{k/48}$, then it is in fact possible to efficiently find a satisfying assignment. Beck's bound was improved by many works over a period of nearly 20 years (e.g.~\cite{alon1991parallel, molloy1998further, srinivasan2008improved, moser2009constructive}) culminating in the landmark work of Moser and Tardos \cite{moser2010constructive}, which gives an efficient algorithmic proof of \cref{thm:local-lemma-general} provided further that one is in the variable setting and some other technical assumptions are satisfied. There has been much work on extending the result of Moser and Tardos to more general settings and the algorithmic aspects of the LLL remain an active area of research (see, e.g., \cite{achlioptas2019beyond} and the references therein).\\

In this work, we are concerned with the following.
\begin{problem}
\label{problem:main}
Suppose that conditions similar to the LLL are satisfied. Can we approximately count the total number of satisfying assignments? Can we sample from approximately the uniform distribution on satisfying assignments?
\end{problem}
 This problem has attracted much attention in the past five years. Below, we only discuss results for approximate counting, noting that similar results also hold for approximate sampling. \\  

In \cite{bezakova2019approximation}, Bez\'akov\'a et al.~ showed that if $\mbf{P} \neq \mbf{NP}$, then it is \emph{not} possible to efficiently approximately count solutions of a Boolean $k$-CNF formula in which every variable is allowed to be present in $d$  constraints for $d \geq 5\cdot 2^{k/2}$, even when the $k$-CNF formula is monotone. For monotone $k$-CNFs, Hermon, Sly, and Zhang \cite{hermon2019rapid} showed that the Glauber dynamics mix rapidly for $d \leq c2^{k/2}$, thereby providing an approximate counting algorithm within a constant factor of the hard regime. For not necessarily monotone $k$-CNFs, Moitra \cite{moitra2019approximate} provided a novel method to \emph{deterministically} approximately count satisfying assignments for $d \lesssim 2^{k/60}$ (where $\lesssim$ hides polynomial factors in $k$), which runs in polynomial time for $k = O(1)$. Using a Markov chain on a certain ``projected space'' inspired by Moitra's method, Feng, Guo, Yin, and Zhang \cite{feng2020fast} relaxed the restriction to $d \lesssim 2^{k/20}$ and removed the requirement $k = O(1)$, although their algorithm is not deterministic. We also mention here the work of Guo, Jerrum, and Liu \cite{guo2019uniform} on ``partial rejection sampling''. 
For $k$-CNFs, their method allows one to \emph{perfectly} sample from the uniform distribution on satisfying assignments, either for ``extremal formulas'' (and $d$ in the LLL regime), or for  formulas for which the intersections between the constraints satisfy some rather stringent size restrictions (and for $d$ matching the hardness regime).\\

Very recently, work of Feng, He, and Yin \cite{feng2020sampling} addressed \cref{problem:main} in the special case where each constraint is violated by a \emph{very small number} of configurations of its variables. Their results are obtained in the following setting. 
\begin{definition}
(1) A constraint satisfaction problem (CSP) is said to be atomic if each constraint is violated by at most one assignment to its variables. 

(2) A $(k,d,q)$-CSP on variables $x_1,\dots, x_n$ is a constraint satisfaction problem in which each $x_i$ takes values in $[q]$, each constraint depends on exactly $k$ variables, and each variable features in at most $d$ constraints. 
\end{definition}

In \cite{feng2020sampling}, a fast randomized algorithm is provided in the setting of \cref{cor:local-lemma-symmetric}, \emph{assuming that the CSP is atomic} and that $pD^{350} \lesssim 1$. For $(k,d,q)$-CSPs which are atomic, they obtain better bounds, leading to an algorithm for $k$-CNF Boolean formulas with $d \lesssim 2^{k/13}$, and an algorithm for proper $q$-colorings of $k$-uniform hypergraphs with $q \gtrsim d^{9/(k-12)}$ (this improves on a previous bound of $q \gtrsim d^{14/(k-14)}$ for $q, k = O(1)$ due to Guo, Liao, Lu, and Zhang \cite{guo2019counting}, although \cite{guo2019counting} provides a deterministic algorithm). 

For not-necessarily-atomic CSPs, \cite{feng2020sampling} simply decompose each constraint into atomic constraints, which leads to the restriction
\begin{align}
\label{eqn:bound-fhy}
p(DN)^{350} \lesssim 1,
\end{align}
where $N$ is an upper bound on the number of violating assignments to the variables of any constraint (hence, $N = 1$ for an atomic CSP). To see that this condition is vastly more restrictive than \cref{cor:local-lemma-symmetric}, consider the case of Boolean CSPs for which each constraint depends on at most $k$ variables. Then, $p \geq 2^{-k}$ so that \cref{eqn:bound-fhy} fails to be applicable as soon as $N \gtrsim 2^{k/350}$. In contrast, \cref{cor:local-lemma-symmetric} shows that a solution exists, provided that $D \lesssim 2^{k}/N$ for all $N \lesssim 2^{k}$. 

The restriction $N \gtrsim 2^{k/350}$ arising from \cref{eqn:bound-fhy} is rather undesirable, since in many applications of the LLL (cf.~\cite{alon2004probabilistic}), $N = \Theta(2^{ck})$ with the constant $c \in (0,1)$ coming from various concentration inequalities. One of the main open problems mentioned in the works \cite{guo2019counting,feng2020sampling} is whether one can go beyond the ``atomic CSP'' framework to provide an affirmative answer to \cref{problem:main} for general CSPs. 

\subsection{Our results} 
We provide, for the first time, approximate counting and sampling algorithms for general CSPs under LLL-like conditions. 
\begin{theorem}
\label{thm:counting}
Let $\Phi = (V,\mc{C})$ denote a constraint satisfaction problem on variables $v_1,\dots, v_n$ and constraints $C_1,\dots, C_m$. For each constraint $C \in \mc{C}$, let $\on{vbl}(C) \subseteq V$ denote the variables it depends on. 
Suppose there exist $p \in (0,1)$, $\Delta\geq 1$, $k\geq 1$, $q\geq 1$ satisfying the following conditions. 
\begin{itemize}
    \item The domain of each variable $v_1,\dots, v_n$ is of size at most $q$.
    \item For each $C \in \mc{C}$, $|\on{vbl}(C)| \leq k$.
    \item For each $C \in \mc{C}$, $\mb{P}[C] \leq p$.
    \item For each $i \in [m]$, $\#\{j\in [m]\setminus \{i\}: \on{vbl}(C_j)\cap \on{vbl}(C_i)\neq \emptyset\} \leq \Delta$, and 
    \item $q^{3}\cdot k\cdot p\cdot \Delta^{7} \leq c$, where $c$ is an absolute constant.  
\end{itemize}
Then, for any $\varepsilon \in (0,1)$, the number of satisfying assignments of $\Phi$ can be deterministically approximated to within relative error $(1\pm \varepsilon)$ in time 
\[\left(\frac{n}{\varepsilon}\right)^{\on{poly}(k,\Delta, \log q)}.\]
\end{theorem}
\begin{remark} 
(1) For Boolean $k$-CNFs with $k = O(1)$, this provides a deterministic approximate counting algorithm for $\Delta \lesssim 2^{k/7}$, where $\lesssim$ hides polynomial factors in $k$. As mentioned above, the previously best known algorithm, either randomized or deterministic, requires $\Delta \lesssim 2^{k/13}$ \cite{feng2020sampling}.

(2) For properly $q$-coloring $k$-uniform hypergraphs, this provides a deterministic approximate counting algorithm for $q \gtrsim \Delta^{7/(k-4)}$. The previous best known algorithms required $q \gtrsim \Delta^{9/(k-12)}$ \cite{feng2020sampling} or $q \gtrsim \Delta^{14/(k-14)}$ for deterministic algorithms \cite{guo2019counting}.
\end{remark}

The framework for proving \cref{thm:counting} also lends itself naturally to an approximate sampling algorithm. 
\begin{theorem}
\label{thm:sampling}
Under the same conditions as \cref{thm:counting} and for any $\varepsilon \in (0,1)$, there is a randomized algorithm to sample from a distribution which is $\varepsilon$-close in total variation distance to the uniform distribution on satisfying assignments of $\Phi$. The running time of the algorithm is  $\left(\frac{n}{\varepsilon}\right)^{\on{poly}(k, \Delta, \log{q})}$. 
\end{theorem}

\subsection{Techniques} 
The works \cite{beck1991algorithmic, alon1991parallel, molloy1998further, srinivasan2008improved} on the algorithmic local lemma predating Moser's work \cite{moser2009constructive} employ the following two step strategy: one first finds a ``good'' partial assignment with the property that the residual formula ``factorizes'' into logarithmic sized components, at which point, one can extend the partial assignment to a complete satisfying assignment efficiently using exhaustive enumeration. For sampling from (approximately) the uniform distribution on satisfying assignments, one therefore only needs to generate these initial partial assignments according to the correct distribution. To accomplish this, we use a generalization of a linear program introduced by Moitra \cite{moitra2019approximate} to approximate the marginal distribution (induced by the uniform distribution on satisfying assignments) of an unassigned variable, conditioned on partial assignments satisfying certain conditions.   

The key conceptual contribution of our work is the following insight: the combinatorial conditions guaranteeing the factorization of the residual formula into logarithmic sized components are essentially the same as those ensuring that the LP can be solved efficiently (see the proof of \cref{lem:LP-error,lem:LP-error-2}). This provides a unifying view of the works \cite{beck1991algorithmic, alon1991parallel, molloy1998further, srinivasan2008improved} on the algorithmic LLL, and the recent works \cite{moitra2019approximate, guo2019counting}, and in our opinion, sheds considerable light on the latter two works.  

From a technical viewpoint, our main contribution is a considerable generalization, refinement, and simplification of the framework introduced by Moitra \cite{moitra2019approximate}. In particular, we eliminate any need to use the algorithmic local lemma of Moser and Tardos \cite{moser2010constructive} (which was an essential ingredient in \cite{moitra2019approximate,guo2019counting} and leads to additional degradation in the quantitative bounds); instead, we show how to efficiently exploit  a certain ``pseudo-random property'' of the initial partial assignment in a direct manner to remove this loss (\cref{lem:LP-error-aggregate}). Our framework also treats approximate counting and approximate sampling on the same footing in a very simple manner, whereas the previous works \cite{moitra2019approximate, guo2019counting} suffered from additional losses in going from approximate counting to approximate sampling. 

We believe that our analysis lays bare the limits of this approach towards approximate counting and sampling. The term $\Delta^7$ in \cref{thm:counting,thm:sampling} comes from the aggregation of two sources of slack. The first is the use of $\{2,3\}$-trees as in \cite{alon1991parallel, molloy1998further, srinivasan2008improved} -- even for the algorithmic LLL, $\{2,3\}$-trees only lead to $\Delta^{4}$ instead of $\Delta$, and being able to use ``denser witness trees'' was the major innovation in the works of Moser \cite{moser2009constructive} and Moser and Tardos \cite{moser2010constructive}. Another reason for the slack is a certain ``factorization property'' required to even write down the linear program efficiently. We believe that with some additional ideas, this second source of slack may be overcome, and leave this as an interesting direction for future research.

\subsection{Organization} In \cref{sec:prelims}, we collect some preliminaries. In \cref{sec:algorithm}, we present a slightly simpler algorithm which proves \cref{thm:counting} with $\Delta^{7}$ replaced by $\Delta^{10}$; the analysis of a key step in this algorithm is completed in \cref{sec:estimation-marginal}. The ideas introduced in \cref{sec:algorithm,sec:estimation-marginal} are further refined in \cref{sec:proofs} to prove \cref{thm:counting} in \cref{sec:counting} and \cref{thm:sampling} in \cref{sec:sampling}.

\section{Preliminaries}\label{sec:prelims}
\subsection{Lov\'asz Local Lemma} 
As mentioned in the introduction, the LLL provides a sufficient condition guaranteeing that the probability of avoiding a collection $\mc{C}$ of ``bad events'' in a probability space is positive. In particular, when the LLL condition \cref{eqn:LLL-condition} is satisfied, the so-called LLL distribution,
\[\mu_{S}[\cdot] := \mb{P}[\cdot \mid \wedge_{C \in \mc{C}}\ol{C}]\]
is well-defined (here, the subscript $S$ is chosen to represent ``satisfying''). The LLL distribution is the central object of study in this paper. We begin by recording a standard comparison between the LLL distribution $\mu_{S}[\cdot]$ and the product distribution on the variables $\mb{P}[\cdot ]$. 
\begin{theorem}[cf.~{\cite[Theorem~2.1]{haeupler2011new}}]
\label{thm:hss-local-lemma}
Under the conditions of \cref{thm:local-lemma-general}, for any event $B$ in the probability space, 
\[\mu_{S}[B] \leq \mb{P}[B]\prod_{C \in \Gamma(B)}(1-x(C))^{-1}.\]
\end{theorem}

\begin{remark}
The above comparison is one-sided, as it ought to be, since for any $C \in \mc{C}$, $\mu_{S}[C] = 0$ while $\mb{P}[C]$ may be positive. 
\end{remark}

For the remainder of this paper, we will restrict ourselves to the variable-version symmetric setting described in \cref{cor:local-lemma-symmetric}, in which case, we choose $x(C) = e\cdot \mb{P}[C]$ for all $C \in \mc{C}$, with $e$ the base of the natural logarithm. 



\subsection{\{2,3\}-trees}
One of the key tools in our analysis will be the notion of $\{2,3\}$-trees, which goes back to Alon's work on the algorithmic local lemma \cite{alon1991parallel}.
\begin{definition}
\label{def:2-3-tree}
Let $G = (V,E)$ be a graph and let $\on{dist}_G(\cdot, \cdot)$ denote the graph geodesic distance. A $\{2,3\}$-tree is a subset of vertices $T \subseteq V$ such that
\begin{itemize}
    \item for any $u,v \in T$, $\on{dist}_G(u,v) \geq 2$;
    \item if one adds an edge between every $u,v \in T$ such that $\on{dist}_G(u,v) = 2\text{ or }3$, then $T$ is connected. 
\end{itemize}
\end{definition}

The next lemma bounds the number of $\{2,3\}$-trees of a given size in terms of the maximum degree of the graph. 

\begin{lemma}[cf.~{\cite[Lemma~2.1]{alon1991parallel}}]
\label{lem:number-2-3-trees}
Let $G = (V,E)$ be a graph with maximum degree $d$. Then, for any $v\in V$, the number of $\{2,3\}$-trees in $G$ of size $t$ containing $v$ is at most $(ed^{3})^{t-1}/2$.
\end{lemma}

Before stating the next lemma, we need some notation. Let $H = (V,E)$ be a hypergraph. Let $\on{Lin}(H)$ denote its line graph i.e.~$V(\on{Lin}(H)) = E$ and there is an edge between $u\neq v \in V(\on{Lin}(H))$ if and only if the hyperedges $u, v \in E$ share a vertex in $V$. Finally, let $L^2(H)$ denote the graph with the same vertex set as $\on{Lin}(H)$ and with an edge between two vertices $u\neq v \in V(L^2(H))$ if and only if $\on{dist}_{\on{Lin}(H)}(u,v)\leq 2$. Then, a simple greedy argument shows the following. 
\begin{lemma}[cf.~{\cite[Lemma~14]{guo2019counting}}]
\label{lem:large-2-3-tree}
Let $H = (V,E)$ be a hypergraph such that each hyperedge in $E$ intersects at most $d$ other hyperedges (equivalently, the degree of $\on{Lin}(H)$ is at most $d$). Let $B\subseteq E(H)$ be a collection of hyperedges which induce a connected subgraph in $L^2(H)$. Then, for any $e^* \in B$, there exists a $\{2,3\}$-tree $T\subseteq B$ in $\on{Lin}(H)$ such that $e^* \in T$ and $|T|\geq |B|/d$.  
\end{lemma}

\section{A simpler algorithm for a more restrictive regime}
\label{sec:algorithm}
In this section and the next one, we present a simpler algorithm which proves a version of \cref{thm:counting} provided that $p \le (10^{5}q^{3}k\Delta^{10})^{-1}$. The design and analysis of this algorithm already contains the basic ideas. Later, in \cref{sec:proofs}, we will introduce some key additional ingredients to refine this algorithm and its analysis in order to prove \cref{thm:counting,thm:sampling}.\\   

Throughout this section and the next one, we fix an arbitrary ordering of the variables $v_1,\dots, v_n \in V$ and an arbitrary ordering of the constraints $C_1,\dots,C_m \in \mc{C}$. Moreover, for notational convenience, we will assume that the domain of each variable is $[q]$; a straightforward modification of the proof shows that we only need the size of the domains to be bounded above by $q$.  
\subsection{Step 1: Finding a guiding assignment}\label{subsec:Step-1}
The goal of this step is to find a partial assignment of the variables, which will serve as a ``guide'' for the rest of the algorithm. 
This step is very much inspired by analogous routines for the algorithmic local lemma \cite{beck1991algorithmic, alon1991parallel, molloy1998further}, and generalizes a similar step in the works \cite{moitra2019approximate, guo2019counting}. We note that in the works \cite{moitra2019approximate, guo2019counting}, it is critical that one is able to efficiently find a partial assignment satisfying each constraint without assigning values to too many variables in any constraint -- while this is indeed possible in the special settings considered in these works, in our setting, such a partial assignment need not exist. The key result of this subsection is \cref{lem:derandomize}. We will first present a randomized construction achieving the guarantees of \cref{lem:derandomize}, and then show how to derandomize it using standard techniques.\\   

We consider the following randomized greedy procedure to construct a partial assignment of $v_1,\dots,v_n$, where $p' > 0$ is a parameter which will be specified later. 
\begin{enumerate}[(R1)]
\item Let $A_{0} = V$ denote the set of initially ``available variables'' and let $F_0 = \emptyset$ denote the set of initially ``frozen variables''. Initialize the stage to $i = 1$. 
\item Select the first variable (according to our order) in $A_{i-1}$.  Denote this variable by $v_i^*$. If no such variable exists, terminate the process.
\item Assign $v_i^*$ a uniformly random value from its alphabet $[q]$. Let $P_i$ denote the partial assignment resulting after the assignment to $v^*_i$.
\item Let  
$$\mc{F}_i = \{j \in [m]: \mb{P}[C_j \mid P_i] > p'\}$$ denote the set of ``dangerous constraints'' under $P_i$. We ``freeze'' all the variables involved in any of the ``dangerous constraints'', i.e.~we set $$F_{i} = F_{i-1} \cup \bigcup_{j \in \mc{F}_i}((\on{vbl}(C_j) \cap A_{i})\setminus \{v_i^*\})\text{ and }A_i = A_{i-1} \setminus (F_{i} \cup \{v^*_i\}).$$
\item Increment $i$ by $1$ and return to (R2).
\end{enumerate}
Note that the process terminates at the first stage $s$ satisfying $A_{s} = \emptyset$. Let $P_1,\dots, P_s$ denote the partial assignments generated during the course of the process. 
Let $\mc{F} = \cup_{i=1}^{s}\mc{F}_i$ denote the set of constraints declared ``dangerous'' at any point during the process. Finally, let $\nu$ denote the distribution on partial assignments given by the random partial assignment $P_s$. We emphasize that $s$ itself is a random variable. \\ 

The following simple observation will be useful later. 
\begin{lemma}
\label{lem:prob-partial}
For all $i \in [s]$ and for all $j \in [m]$, 
$$\mb{P}[C_{j} \mid P_i] \leq p'q.$$
\end{lemma}
\begin{proof}
Fix $j \in [m]$. If $j \notin \mc{F}$, then we are done, so assume that $j \in \mc{F}$. Let $i$ be the first stage for which $C_j \in \mc{F}_i$. Then, $\mb{P}[C_{j} \mid P_{i-1}] \leq p'$, and since $P_i$ extends $P_{i-1}$ by assigning $v_{i}^*$, we have
\begin{align*}
    \mb{P}[C_j \mid P_i] 
    &\leq \frac{\mb{P}[C_j \mid P_{i-1}]}{\min_{u \in [q]}\mb{P}[v_i^* = u]}\\
    &\leq p'q.
\end{align*}
Note that if $C_j \in \mc{F}_i$, all variables in $\on{vbl}(C_j)$ are added to $F_i$. Hence, no variable in $\on{vbl}(C_j)$ is assigned a value during the remainder of the process so that $\mb{P}[C_j \mid P_{i+k}] = \mb{P}[C_j \mid P_{i}]$ for all $0\leq k \leq s-i$. 
\end{proof}






For $h \in [n]$, let $\mf{S}(h)$ be the $\sigma$-algebra generated by the output of the procedure on partial assignments of $v_1,\dots,v_h$. Also, for each $h \in [n]$, let $\iota(h) = \max \{i \in [h] : v_i^* = v_{h'}\text{ for some }h' \leq h\}$. In other words, $\iota(h)$ denotes the number of variables in $v_1,\dots, v_h$ which are assigned values by the partial assignment. 
\begin{lemma}\label{lem:martingale}
    Let $T \subseteq \mc{C}$ be a collection of constraints such that for any $C, C' \in T$ with $C \neq C'$, $\on{vbl}(C) \cap \on{vbl}(C') = \emptyset$. Then, letting
    \[
        M_h = \prod_{C\in T} \mb{P}[C \mid P_{\iota(h)}],
    \] 
    we have
    \[\mb{E}_\nu[M_{h+1} \mid \mf{S}(h)] = M_h.\]
\end{lemma}
\begin{proof}
Since the variables $v_1,\dots, v_n$ are processed in order, given $\mf{S}(h)$, it is determined whether $v_{h+1}$ is frozen or not. If $v_{h+1}$ is frozen, then $M_{h+1}=M_h$. 

Otherwise, $v_{h+1}$ is not frozen. Note that, in this case, $\iota(h+1) = \iota(h) + 1$ and $v^*_{\iota(h+1)} = v_{h+1}$. Since $\on{vbl}(C)$ are disjoint for $C\in T$, there is at most one constraint $C_T\in T$ such that $v_{h+1} \in \on{vbl}(C_T)$. If there is no such constraint $C_T$, then $M_{h+1}=M_h$. 

Consider now the case that $v_{h+1}$ is not frozen, and there exists a unique constraint $C_T \in T$ such that $v_{h+1}\in \on{vbl}(C_T)$. We have $\mb{P}[C \mid P_{\iota(h+1)}] = \mb{P}[C \mid P_{\iota(h)}]$ for all $C\ne C_T$. Moreover, since $v_{h+1}$ is assigned a uniformly distributed value in $[q]$, we have 
    \begin{align*}
    \mb{E}_\nu\bigg[\mb{P}[C_T \mid P_{\iota(h+1)}] \mid  \mf{S}(h)\bigg] 
    &= \frac{1}{q}\sum_{a\in [q]} \frac{\mb{P}[C_T \wedge P_{\iota(h)} \wedge v_{h+1}=a]}{\mb{P}[P_{\iota(h)} \wedge v_{h+1}=a]}\\
    &= \frac{\mb{P}[C_T,P_{\iota(h)}]}{\mb{P}[P_{\iota(h)}]}\\
    &= \mb{P}[C_T \mid P_{\iota(h)}].
\end{align*}
Combining these cases gives the desired conclusion.
\end{proof}

Since $\mb{P}[C] \leq p$ for all $C \in \mc{C}$, the following is an immediate corollary of \cref{lem:martingale}.
\begin{corollary} \label{cor:mu_S-tree}
     Let $T \subseteq \mc{C}$ be a collection of constraints such that for any $C, C' \in T$ with $C \neq C'$, $\on{vbl}(C) \cap \on{vbl}(C') = \emptyset$. Then, 
    \[
        \mb{E}_\nu\left[\prod_{C\in T} \mb{P}[C \mid P_{s}]\right] \le p^{|T|}.
    \]
\end{corollary}

We are now ready to state and prove the key property of guiding partial assignments, which is that they ``factorize'' the constraint satisfaction problem into small connected components. Let $H = (V, \mc{C})$ denote the hypergraph induced by the constraint satisfaction problem and let $G(\mc{C})$ denote its line graph. Let $G^2(\mc{F}) = (\mc{F}, E)$ denote the graph whose vertices are $\mc{F}$, and $C_i \neq C_j \in \mc{F}$ are adjacent if and only if $\on{dist}_{G(\mc{C})}(C_i, C_j)\leq 2$. 

\begin{lemma}
\label{lem:small-comp}
Let $p\leq p'/(10\Delta^3)$.
The $\nu$-probability that $G^2(\mc{F})$ has a connected component of size at least $L$ is at most $n\Delta\cdot 2^{-L/\Delta}$. 
\end{lemma}
\begin{proof}
Suppose that $\mf{C}$ is a connected component in $G^2(\mc{F})$ with $|\mf{C}| \geq L$. By \cref{lem:large-2-3-tree}, there exists a $\{2,3\}$-tree $T$ in $G(\mc{C})$ with $|T|\geq |\mf{C}|/\Delta \geq L/\Delta$. Let $\mf{T}$ denote the set of $\{2,3\}$-trees in $G(\mc{C})$ of size $L/\Delta$. Then,
\begin{align*}
    \nu[G^2(\mc{F}) \text{ has a connected component of size}\geq L]
    &\leq \nu[G^2(\mc{F})\text{ contains a }\{2,3\}\text{-tree of size}\geq L/\Delta]\\
    &\leq \sum_{T \in \mf{T}}\nu[T \subseteq \mc{F}]\\
    &\leq \sum_{T \in \mf{T}}(p/p')^{L/\Delta}\\
    &\leq |\mc{F}|(e\Delta^{3})^{L/\Delta}(p/p')^{L/\Delta}\\
    &\leq n\Delta (e\Delta^{3}p/p')^{L/\Delta}\\
    &\leq n\Delta\cdot 2^{-L/\Delta}.
\end{align*}
The fourth line uses \cref{lem:number-2-3-trees} and the last line uses the assumed bound on $p/p'$. Let us explain the third line. By \cref{cor:mu_S-tree} and Markov inequality, for any $T \in \mf{T}$, 
\[
\nu[T \subseteq \mc{F}] \leq \nu\left[\left(\prod_{C \in T}\mb{P}[C \mid P_s]\right) > p'^{|T|}\right] \leq \left(\frac{p}{p'}\right)^{|T|} \leq \left(\frac{p}{p'}\right)^{L/\Delta}. \qedhere
\]
\end{proof}

The preceding lemma shows that with high probability, our random greedy process returns a partial assignment satisfying the condition in \cref{lem:small-comp} (for $L$ sufficiently large). Using the standard method of conditional expectations, we can find such a partial assignment deterministically.  

\begin{proposition}
\label{lem:derandomize}
There exists a deterministic algorithm running in time $O(n^{\on{poly}(\log q, \log \Delta, k)})$ which generates a sequence of partial assignments $P_1,\dots, P_s$ with the following properties. 
\begin{enumerate}
    \item For all $i \in [s]$, $P_i$ assigns values to $i$ variables, and $P_i$ extends $P_{i-1}$. 
    \item $A_{s} = \emptyset$.
    \item For all $i \in [s]$ and $j \in [m]$, $\mb{P}[C_j \mid P_i] \leq p'q$. 
    \item Every connected component in $G^2(\mc{F})$ has size at most  $L = 10\Delta\log(\Delta n)$.
\end{enumerate}
\end{proposition}
\begin{proof}
Let $L' = 10\log(\Delta n)$, and let $\mf{T}$ denote the collection of all \{2,3\}-trees of size $L'$ in $G(\mc{C})$. Note that $|\mf{T}| \leq \on{poly}(n^{\log^{2}\Delta})$, and indeed, it is easily seen (cf.~\cite{alon1991parallel}) that the collection $\mf{T}$ can be constructed in time $\on{poly}(n^{\log^{2}\Delta})$.

Now, for a partial assignment $X$, define
\[H(X) = \sum_{T \in \mf{T}}\prod_{C \in T} \left(\mb{P}[C \mid X]/p'\right).\]
By the proof of \cref{lem:small-comp}, if we can find a sequence of partial assignments $P_1,\dots, P_s$ satisfying properties (1), (2), (3) such that $H(P_s) < 1$, then (4) is also satisfied, since any $\{2,3\}$-tree in $G(\mc{C})$ of size $L'$ contributes at least $1$ to the sum.  

To find such a sequence of partial assignments, we follow the same greedy procedure as before, except now, after having chosen $P_{i-1}$ and $v_i^*$, we choose the value of $v_i^*$ to be
\[\argmin_{a \in [q]}H(P_{i-1} \wedge {v_i^* = a}).\]
We claim that $H(P_{i}) \leq H(P_{i-1})$ for all $i \in [s]$. Indeed, for every $T \in \mf{T}$, there exists at most one $C_T \in T$ such that $v_i^* \in \on{vbl}(C_T)$. Therefore,
\begin{align*}
    \sum_{a\in [q]}\mb{P}[v_i^*=a]H(P_{i-1} \wedge v_i^* = a) 
    &= \sum_{T \in \mf{T}}\sum_{a\in [q]}(p')^{-1}\mb{P}[C_T\mid P_{i-1}\wedge v_i^* = a]\mb{P}[v_i^* = a]\prod_{C \in T\setminus C_T}\frac{\mb{P}[C \mid P_{i-1}]}{p'}\\
    &= H(P_{i-1}),
\end{align*}
which shows that it is possible to choose $P_i$ to ensure $H(P_{i}) \leq H(P_{i-1})$. Finally, since
$$H(\emptyset) \leq (n\Delta)\cdot (e\Delta^{3})^{L'}\cdot (p/p')^{L'} < 1,$$
we are done. 
\end{proof}








\subsection{Step 2: Approximate counting}

Let $P_0 = \emptyset$ and $P_1,\dots, P_s$ denote the sequence of partial assignments returned by \cref{lem:derandomize}. As before, we will denote the vertices which are successively assigned values by $v_1^*,\dots, v_s^*$ and we will denote their values under $P_s$ by $a_1^*,\dots, a_s^*$. 

For a partial assignment $X$, let $\mc{S}_X$ denote the number of (complete) satisfying assignments extending $X$. We will use $\mc{S}_{P_0}$ (or $\mc{S}_{\emptyset}$) to denote the set of all complete satisfying assignments. Then,
\begin{align*}
    \frac{|\mc{S}_{P_s}|}{|\mc{S}_{P_0}|} 
    &= \frac{|\mc{S}_{P_1}|}{|\mc{S}_{P_0}|}\cdot \frac{|\mc{S}_{P_2}|}{|\mc{S}_{P_1}|}\cdots \frac{|\mc{S}_{P_s}|}{|\mc{S}_{P_{s-1}}|} \\
    &= \prod_{i=1}^{s}\mu_{S}[v_i^* = a_i^* \mid P_{i-1}],
\end{align*}
where recall that $\mu_{S}$ denotes the uniform measure on all satisfying assignments i.e.~on $\mc{S}_{P_0}$. Thus, to approximate $|\mc{S}_{P_0}|$, it suffices to approximate $|\mc{S}_{P_s}|$ and $\mu_{S}[v_i^* = a_i^* \mid P_{i-1}]$ for all $i\in [s]$. 

\begin{lemma}
\label{lem:initial-factorization}
For $P_s$ returned by \cref{lem:derandomize}, $|\mc{S}_{P_s}|$ can be computed exactly in time $n^{\on{poly}(\Delta, k, \log q)}$. 
\end{lemma}
\begin{proof}
Since $A_{s} = \emptyset$, the set of variables left unassigned by $P_{s}$ is precisely $F_{s}$. 
Let $G_1,\dots, G_{r}$ denote the maximal connected components of $G^{2}(\mc{F})$ and let $V'_{1},\dots, V'_{r}$ denote the variables appearing in any constraint in $G_1,\dots, G_{r}$. 
By the maximality of $G_1,\dots, G_r$, the sets $V_1',\dots, V_s'$ are mutually disjoint.
Also, by the maximality of $G_1,\dots, G_{r}$, there does not exist any $C \in \mc{C}$ such that $\on{vbl}(C) \cap V'_i \neq \emptyset$ and $\on{vbl}(C) \cap V'_j \neq \emptyset$ for some $i \neq j \in [r]$, since otherwise, some vertex in $G_i$ would be connected in $G^2(\mc{F})$ to some vertex in $G_j$. Finally, note that $|G_i| \leq L$, and hence $|V'_i| \leq kL$ for all $i \in [r]$.

Since any $v \in F_{s}$ must belong to some $C \in \mc{F}$ and since any $C \in \mc{F}$ must belong to some $G_i$, it follows that
$$F_{s} \subseteq V_1'\cup \dots \cup V'_s.$$
Moreover, as seen in the previous paragraph, there are no constraints involving variables from both $V'_i$ and $V'_j$ for $i\neq j$. Therefore, for each $i \in [r]$, we can exhaustively enumerate all assignments to $V'_i \cap F_s$, check how many of them satisfy all relevant constraints, and finally take the product over $i \in [r]$ in the claimed time.   
\end{proof}

Approximating $\mu_{S}[v_i^* = a_i^* \mid P_{i-1}]$ for $i \in [s]$ is much more involved and will be the content of the next section. Let $\delta \in (0,1)$ and $q^{-n} \leq r_- \leq r_+ \leq q^{n}$ be parameters. In \cref{prop:lp-certify}, we will construct a subroutine $\on{Alg}_{r_-, r_+, \delta}$ with the following properties. Suppose $p' \leq (10000q^{3}k\Delta^{7})^{-1}$ and let $b \in [q]$.
\begin{itemize}
    \item $\on{Alg}_{r_-, r_+, \delta}$ runs in time $\on{poly}(n,k,q)\cdot 2^{\log(1/\delta)\cdot\on{poly}(\Delta, k, \log{q})}$.
    \item $\on{Alg}_{r_-, r_+, \delta}$ returns $\on{YES}$ if and only if 
    \[r_-(1-\delta)\leq \frac{\mu_S[v_i^* = b\mid P_{i-1}]}{\mu_S[v_i^* = a_i^* \mid P_{i-1}]}\leq r_+(1+\delta).\]
\end{itemize}

Let $\varepsilon \in (0,1)$ be a parameter. Then, using such a subroutine along with binary search on the parameters $r_-, r_+$, we can clearly approximate $\mu_S[v_i^* = a_i^* \mid P_{i-1}]$ up to a multiplicative factor of $\exp(\varepsilon/n)$ for each $i \in [s]$ in time $(n/\varepsilon)^{\on{poly}(\Delta, k, \log{q})}$. Together with \cref{lem:initial-factorization}, this therefore provides an approximation of $|\mc{S}_{P_0}|$ up to relative error $\exp(\varepsilon)$. 

\section{Efficient estimation of the marginals} \label{sec:estimation-marginal}
We will continue to use the notation and conventions of the previous section. Throughout, we fix a partial assignment $P_s$ as returned by \cref{lem:derandomize}. By considering the fixed order $v_1,\dots, v_n$ of the variables, this fixes the identity of the variables $v_1^*,\dots, v_s^*$ as well as the intermediate sequence of partial assignments $P_1,\dots, P_{s-1}$. Throughout, we also fix $\ell \in [s]$. Our goal is to efficiently approximate the conditional probabilities $p_{\ell}(a):= \mb{P}[v_\ell^* = a \mid P_{\ell-1}]$ for all $a \in [q]$. We will use $\mu_S$ to denote the uniform measure over all (complete) satisfying assignments, and for a partial assignment $x$, $\mu_S [\cdot \mid x]$ to denote the uniform measure on all (complete) satisfying assignments extending $x$. For partial assignments $x, x'$, the notation $x' \to x$ means that $x$ is an extension of $x'$ (i.e.~each variable that is assigned in $x'$ is also assigned in $x$ to the same value). Finally, we emphasize that $\mb{P}[\cdot]$ will always mean the product measure on the variables $v_1,\dots, v_n$.

\subsection{Idealized coupling procedure and the idealized decision tree}
Let $p''>0$ be a parameter which will be chosen later. Fix $a\neq b \in [q]$. Let $P_{\ell}(a)$ denote the partial assignment extending $P_{\ell-1}$ obtained by setting $v_\ell^* = a$ and let $P_\ell(b)$ be defined analogously. We begin by describing a coupling between assignments extending $P_{\ell}(a)$ and $P_{\ell}(b)$, which will motivate subsequent discussion. We note that this coupling is not meant to actually be implemented by the algorithm.

\begin{enumerate}[(C1)]
    \item Initialize the partial assignments $X = P_{\ell}(a)$ and $Y = P_{\ell}(b)$. Initialize $(V_S)_{X,Y} = \{v_1^*,\dots, v_\ell^*\}$ (the collection of ``set'' variables) and $(V_D)_{X,Y} = \{v_\ell^*\}$ (the collection of ``dangerous'' variables). 
    \item Choose the lowest numbered constraint $A \in \mc{C}$ such that $(V_D)_{X,Y} \cap \on{vbl}(A) \neq \emptyset$  and $ \on{vbl}(A)\cap ((V_D)_{X,Y} \cup (V_S)_{X,Y})^{c} \neq \emptyset$. If no such $A \in \mc{C}$ exists, then terminate. 
    \item Choose the lowest numbered variable $v \in \on{vbl}(A) \cap ((V_D)_{X,Y} \cup (V_S)_{X,Y})^{c}$. 
    \item Sample a pair of values $(v_X, v_Y)$ according to the maximal coupling of the marginal distribution of $\mu_S$ at $v$, conditioned on $X$ and $Y$ respectively. 
    \item Update $X$ by assigning $v = v_X$, and update $Y$ by assigning $v = v_Y$. Update $(V_S)_{X,Y}$ by adding $v$.  
    \item Let $D_{X,Y} = \{u \in (V_{S})_{X,Y} : X(u) \neq Y(u)\}$. Let $\mc{F}_{X,Y} = \{C \in \mc{C}$ : $\mb{P}[C \mid X] > p''$ or $\mb{P}[C \mid Y] > p''\}$. Update $$(V_D)_{X,Y} = D_{X,Y} \cup \bigcup_{C \in \mc{F}_{X,Y}}(\on{vbl}(C)\cap (V_S)_{X,Y}^{c}),$$ 
    and return to (C2).
\end{enumerate}
\bigskip
We record a few simple observations.
\begin{enumerate}[(O1)]
    \item The set $(V_{S})_{X,Y}$ increases throughout the process. 
    \item $\mc{F}_{X,Y}$ is non-decreasing throughout the process. Indeed, once $C \in \mc{F}_{X,Y}$, no other $v \in \on{vbl}(C)$ can be chosen in (C3), so that the conditional probability of $C$ with respect to all subsequent partial assignments remains the same.  
    \item The set $(V_{D})_{X,Y}$ is non-decreasing throughout the process. 
\end{enumerate}
\bigskip 

The above coupling process may be viewed as randomly traversing root-to-leaf trajectories in an idealized deterministic rooted decision tree $\mc{T}$, defined using the following inductive procedure. 

\begin{enumerate}[(T1)]
    \item The root of the tree consists of the partial assignments $(x_0, y_0):= (P_\ell(a), P_\ell(b))$. 
    \item Given a node $(x,y)$ (consisting of partial assignments on the same variables $(V_S)_{x,y}$), construct 
    $D_{x,y}, \mc{F}_{x,y}$ as in (C6). Let 
    $$(V_{D})_{x,y} = D_{x,y} \cup \bigcup_{C \in \mc{F}_{x,y}}(\on{vbl}(C) \cap (V_S)_{x,y}^{c}).$$
    \item If there is no $A \in \mc{C}$ with $(V_{D})_{x,y} \cap \on{vbl}(A) \neq \emptyset$ and $\on{vbl}(A) \cap ((V_D)_{x,y} \cup (V_S)_{x,y})^{c} \neq \emptyset$, then $(x,y)$ is a leaf of $\mc{T}$.
    \item Otherwise, let $A \in \mc{C}$ be the lowest numbered such constraint, and let $v_{x,y}$ be the lowest numbered variable in $\on{vbl}(A) \cap ((V_D)_{x,y} \cup (V_S)_{x,y})^{c}$. The children of $(x,y)$ in $\mc{T}$ consist of all possible extensions of $(x,y)$ obtained by assigning a value to the variable $v_{x,y}$.  
\end{enumerate}

The next lemma collects some useful properties of $\mc{T}$. 

\begin{lemma}
\label{lem:tree-properties}
For $\mc{T}$ as defined above,
\begin{enumerate}
    \item For any node $(x,y) \in \mc{T}$,
    \[\mb{P}[C_j \mid x] \leq p''q\text{ and } \mb{P}[C_j \mid y] \leq p''q\quad \text{ for all }j \in [m].\]
    \item Assuming that $e\cdot p''q\cdot \Delta \le 1$, for any node $(x,y) \in \mc{T}$ and for any $v \notin (V_{D})_{x,y},$
    \begin{align*}
        \on{TV}(\mu_{S}[v = \cdot \mid x], \mb{P}[v = \cdot]) &\leq (1-3p''q)^{-\Delta} - 1,\\
        \on{TV}(\mu_{S}[v = \cdot \mid y], \mb{P}[v = \cdot]) &\leq (1-3p''q)^{-\Delta} - 1.
    \end{align*}
    \item For any leaf $(x,y) \in \mc{T}$, there is a partition $V = (V_{D})_{x,y} \cup (V_{G})_{x,y}\cup (V_R)_{x,y}$ such that every variable in $(V_{G})_{x,y}$ is assigned to the same value by both $x$ and $y$ and such that there is no constraint $C \in \mc{C}$ with variables in both $(V_{D})_{x,y}$ and $(V_R)_{x,y}$. 
\end{enumerate}
\end{lemma}
\begin{proof}
(1) is immediate from (O2) and the same argument as \cref{lem:prob-partial}. 

(3) follows immediately using the termination criterion (T3) by taking $(V_G)_{x,y} = (V_S)_{x,y} \setminus (V_D)_{x,y}$, and $(V_R)_{x,y} = ((V_D)_{x,y} \cup (V_S)_{x,y})^{c}$. 

Finally, for (2), 
setting $\mu(\cdot) = \mu_{S}[v = \cdot \mid x]$ and $\nu(\cdot) = \mb{P}[v = \cdot]$, we get
\begin{align*}
    \on{TV}(\mu_S[v = \cdot \mid x], \mb{P}[v = \cdot])
    &= \on{TV}(\mu, \nu)\\
    &= \sum_{a \in [q]:\mu(a) > \nu(a)}(\mu(a) - \nu(a))\\
    &\leq \sum_{a \in [q] : \mu(a) > \nu (a)}\frac{1}{q}\left((1-3p''q)^{-\Delta} - 1\right)\\
    &\leq (1-3p''q)^{-\Delta} - 1,
\end{align*}
where the second line uses the definition of total variation distance, and the third line uses (1) and \cref{thm:hss-local-lemma}. The same argument works for $y$ as well. 
\end{proof}

The idealized coupling process and decision tree are naturally associated to the following quantities. 

\begin{definition}
For $\mc{T}$ as defined above,
\begin{itemize}
    \item For any node $(x,y) \in \mc{T}$, let $\mc{S}_{x}$ denote the set of (complete) satisfying assignments extending $x$, and let $\mc{S}_{y}$ denote the set of (complete) satisfying assignments extending $y$. 
    \item For any node $(x,y) \in \mc{T}$, let $\mu_{\on{cp}}(x,y)$ denote the probability that the idealized coupling process reaches $(x,y)$. For any $(x,y) \notin \mc{T}$, $\mu_{\on{cp}}(x,y) = 0$.  
    \item For any node $(x,y) \in \mc{T}$, let
    \begin{align*}
        p_{x,y}^{x}
        &= \frac{\mu_{\on{cp}}(x,y)}{\mu_S[x \mid x_0]},\\
        p_{x,y}^{y} 
        &= \frac{\mu_{\on{cp}}(x,y)}{\mu_S [y\mid y_0]}.
    \end{align*}
\end{itemize}
\end{definition}

We conclude this subsection with some simple, but crucial, relations between these quantities.

\begin{lemma}
\label{lem:property-pxy}
For $\mc{T}$, $\mc{S}_{x}, \mc{S}_y, p^{x}_{x,y}, p^{y}_{x,y}$ as above,
\begin{enumerate}
    \item $p_{x,y}^{x}, p_{x,y}^{y} \in [0,1]$. 
    \item $p_{x_0, y_0}^{x_0}, p_{x_0, y_0}^{y_0} = 1$. 
    \item For every non-leaf node $(x,y) \in \mc{T}$ whose children are defined on the set $(V_S)_{x,y} \cup \{v_{x,y}\}$, letting $v = v_{x,y}$ and letting $x_{v(a)}$ denote the extension of $x$ obtained by setting $v$ to $a$ (and similarly for $y_{v(a)}$),
    \begin{align*}
        p^{x}_{x,y} &= \sum_{b\in [q]}p^{x_{v(a)}}_{x_{v(a)},y_{v(b)}}\quad \text{ for all } a \in [q],\\
        p^{y}_{x,y} &= \sum_{b\in [q]}p^{y_{v(a)}}_{x_{v(b)},y_{v(a)}}\quad \text{ for all } a \in [q].
    \end{align*}
    \item For every node $(x,y) \in \mc{T}$, 
    \[\frac{|\mc{S}_x|\cdot p^{x}_{x,y}}{|\mc{S}_y|\cdot p^{y}_{x,y}}=\frac{|\mc{S}_{x_0}| }{|\mc{S}_{y_0}|}.\]
    \item For every non-leaf node $(x,y) \in \mc{T}$ whose children are defined on the set $(V_S)_{x,y} \cup \{v_{x,y}\}$ and for $\eta = (1-3p''q)^{-\Delta}-1$, if $\eta \leq 1/(2q)$, then (letting $v = v_{x,y}$)
    \begin{align*}
     \sum_{b\neq a }p_{x_{v(a)},y_{v(b)}}^{x_{v(a)}} &\leq 4q\eta \cdot p^{x}_{x,y} \quad \text{ for all }a \in [q]. \\
    \sum_{b\neq a} p_{x_{v(b)},y_{v(a)}}^{y_{v(a)}} &\leq 4q\eta \cdot p^{y}_{x,y} \quad \text{ for all }a \in [q].
    \end{align*}
\end{enumerate}
\end{lemma}

\begin{proof}
(1) and (2) are trivial. 

For (3), we note that for any $a \in [q]$,
\begin{align*}
    \sum_{b \in [q]}p^{x_{v(a)}}_{x_{v(a)}, y_{v(b)}}
    &= \frac{\sum_{b\in [q]}\mu_{\on{cp}}(x_{v(a)},y_{v(b)})}{\mu_S[x_{v(a)}\mid x_0]} \\
    &= \frac{\mu_{\on{cp}}(x,y)\cdot \mu_S[v = a \mid x]}{\mu_S[x \mid x_0]\cdot \mu_S[v=a\mid x]}\\
    &= p^{x}_{x,y}.
\end{align*}
The same argument also applies to $p^{y}_{x,y}$.

For (4), we have indeed that
\begin{align*}
    \frac{p^{x}_{x,y}}{p^{y}_{x,y}}&= \frac{\mu_S[y \mid y_0]}{\mu_S[x \mid x_0]}\\
    &= \frac{|\mc{S}_y|/|\mc{S}_{y_0}|}{|\mc{S}_{x}|/|\mc{S}_{x_0}|}.
\end{align*}

Finally, for (5), $\Pi_{x,y}$ denote the optimal coupling of $\mu_S[u = \cdot \mid x]$ and $\mu_S[u = \cdot \mid y]$, we have for any $a \in [q]$, 
\begin{align*}
\frac{\sum_{b\neq a}p^{x_{v(a)}}_{x_{v(a)},y_{v(b)}}}{p_{x,y}^{x}} 
&= \frac{\sum_{b\neq a}\mu_{\on{cp}}(x_{v(a)},y_{v(b)})}{\mu_{\on{cp}}(x,y) \mu_S[v=a\mid x]}\\
&= \frac{\sum_{b\neq a}\Pi_{x,y}(a, b)}{\mu_S[v = a\mid x]}\\
&\leq \frac{\sum_{b\neq a}\Pi_{x,y}(a,b)}{(1/q) - \eta}\\
&\leq 2q\cdot \on{TV}(\mu_S[v=\cdot\mid x], \mu_S[v=\cdot\mid y])\\
&\leq 2q\left(\on{TV}(\mu_S[v=\cdot \mid x], \mb{P}[v=\cdot]) + \on{TV}(\mb{P}[v=\cdot], \mu_S[v=\cdot \mid y])\right)\\
&\leq 4q\eta,
\end{align*}
where the third line follows from (2) of \cref{lem:tree-properties}, the fourth line follows from the characterization of the total variation distance in terms of optimal coupling, and the last line follows again from (2) of \cref{lem:tree-properties}.
\end{proof}

\subsection{Setting up the linear program}\label{subsec:setup-LP}
The most important property of the quantities defined above is (3) in \cref{lem:property-pxy}, which shows that given $|\mc{S}_x|/|\mc{S}_y|, p^{x}_{x,y}, p^{y}_{x,y}$ at any node $(x,y) \in \mc{T}$, one obtains the key quantity $|\mc{S}_{x_0}|/|\mc{S}_{y_0}|$. What makes this property useful is the following observation, which shows that for $(x,y) \in \mc{T}$ \emph{which are leaves}, the ratio $|\mc{S}_{x}|/|\mc{S}_y|$ can be computed efficiently. 

\begin{lemma}
\label{lem:factorization}
For any leaf $(x,y) \in \mc{T}$, $|\mc{S}_x|/|\mc{S}_y|$ can be computed in time $\on{poly}(n,k,q)\cdot q^{|(V_D)_{x,y}|}$.
\end{lemma}
\begin{proof}
Let $(x,y) \in \mc{T}$ be a leaf and let $(V_D)_{x,y}, (V_G)_{x,y}, (V_R)_{x,y}$ be the partition of $V$ as in (3) of \cref{lem:tree-properties}. All the unassigned variables (under $x,y$) are in $(V_D)_{x,y} \cup (V_R)_{x,y}$ and note that there are no constraints with variables in both $(V_D)_{x,y}$ and $(V_R)_{x,y}$. 
Further, the number of ways of assigning values to variables in $(V_R)_{x,y}$ so that they satisfy all constraints with variables in $(V_R)_{x,y} \cup (V_G)_{x,y}$ is the same (and hence, does not contribute to the ratio) since each variable in $(V_G)_{x,y}$ is assigned to the same value by both $x$ and $y$. Therefore, the ratio $|\mc{S}_x|/|\mc{S}_y|$ is equal to the ratio of the number of ways of assigning values to the unassigned variables in $(V_D)_{x,y}$ such that all constraints with variables in $(V_D)_{x,y} \cup (V_G)_{x,y}$ are satisfied. This can be done by exhaustive enumeration in the claimed time. 
\end{proof}

Motivated by the preceding discussion, let $L \geq 2$ be a parameter to be chosen later and consider the $L$-truncated decision tree defined as follows.

\begin{definition}
\label{def:truncated-tree}
For $L \geq 2$ and with $\mc{T}$ as before, we define the $L$-truncated decision tree $\mc{T}_L$ to consist of those nodes $(x,y) \in \mc{T}$ for which $|(V_S)_{x,y}| \leq L + \ell$. We let $\mc{L}_L$ denote the leaves of $\mc{T}_L$, $\mc{L}_L^{g}$ denote those leaves in $\mc{L}_L$ which have $|(V_S)_{x,y}| \leq L+\ell-1$ (in particular, these are also leaves of $\mc{T}$), and let $\mc{L}_L^{b}$ denote the remaining leaves.  
\end{definition}

We now set up a linear program to mimic the quantities $p^{x}_{x,y}, p^{y}_{x,y}$ for each node of $\mc{T}_L$. Formally, given parameters $r_- \leq r_+$, $\eta = (1-3p''q)^{-\Delta}-1$ and $\mc{T}_L$, we check whether the following linear program in variables $\h p_{x,y}^{x}$ and $\h p_{x,y}^{y}$ is feasible:

\begin{enumerate}[(LP1)]
    \item For all $(x,y) \in \mc{T}_L$, $0\leq \h p^{x}_{x,y}, \h p^{y}_{x,y} \leq 1$. 
    \item For every $(x,y) \in \mc{L}_L^{g}$, $$r_- \le \frac{\h p_{x,y}^{x}|\mc{S}_x|}{\h p_{x,y}^{y} |\mc{S}_y|} \le r_+.$$ 
    \item $\h p_{x_0,y_0}^{x_0}=\h p_{x_0,y_0}^{y_0}=1$. Moreover, for every node $(x,y) \in \mc{T}_L \setminus \mc{L}_L$ whose children are defined on the set $(V_S)_{x,y} \cup \{v_{x,y}\}$ and letting $v = v_{x,y}$,  
    \begin{align*}
    \h p_{x,y}^{x} = \sum_{b\in [q]} \h p_{x_{v(a)},y_{v(b)}}^{x_{v(a)}}\quad \text{ for all } a \in [q],\\
    \h p_{x,y}^{y} =  \sum_{b\in [q]} \h p_{x_{v(b)},y_{v(a)}}^{y_{v(a)}}\quad \text{ for all } a \in [q].
    \end{align*}
    \item For every node $(x,y) \in \mc{T}_L \setminus \mc{L}_L$ whose children are defined on the set $(V_S)_{x,y} \cup \{v_{x,y}\}$ and letting $v = v_{x,y}$, for every $a \in [q]$,
    \begin{align*}
        \sum_{b\neq a}\h p_{x_{v(a)},y_{v(b)}}^{x_{v(a)}} &\leq 4q\eta \cdot\h p^{x}_{x,y},\\
         \sum_{b\neq a}\h p_{x_{v(b)},y_{v(a)}}^{y_{v(a)}} &\leq 4q\eta \cdot\h p^{y}_{x,y}.
    \end{align*}
\end{enumerate}

\begin{claim}
\label{claim:lp-feasible}
The above LP is feasible for $r_- = r_+ = |\mc{S}_{x_0}|/|\mc{S}_{y_0}|$.
\end{claim}
\begin{proof}
This follows immediately by taking $\h p^{x}_{x,y} = p^{x}_{x,y}, \h p^{y}_{x,y} = p^{y}_{x,y}$ and using \cref{lem:property-pxy}.
\end{proof}

\begin{claim}
\label{claim:lp-runtime}
For every $r_-, r_+, \eta$ which can be represented in $\on{poly}(n,q)$ bits, the feasibility of the above LP can be checked in time $\on{poly}(n, q^{L})$.
\end{claim}
\begin{proof}
This follows from standard guarantees on the running time of linear programming (cf. \cite{khachiyan1979polynomial}) since the number of variables and constraints in the LP are $O(q^{L+O(1)})$ and since $|\mc{S}_{x}|/|\mc{S}_{y}|$ can be represented using $\on{poly}(n,q)$ bits. \end{proof}

\subsection{Analysis of the linear program}
We now show that the feasibility of the above LP (for sufficiently large $L$ and appropriately chosen $p''$) implies that $r_-$ (respectively $r_+$) is an approximate lower (respectively upper) bound for $|\mc{S}_{x_0}|/|\mc{S}_{y_0}|$. Given this, we will be able to use binary search in order to approximate $|\mc{S}_{x_0}|/|\mc{S}_{y_0}|$. The key point is that the approximation error decays exponentially in $L$, which will allow us to take $L$ small enough to ensure that this procedure is efficient. 

\begin{proposition}
\label{prop:lp-certify}
Let $p'' \leq (100 q^{2}k \Delta^{4})^{-1}$, $p' \leq p''/(100\Delta^{3}q)$, and $L \geq 8k\Delta^{2}$. Then, the feasibility of the above LP with parameters $r_-, r_+$ implies that
\[\left(1-4\cdot 2^{-L/(k\Delta^{2})}\right)r_- \leq\frac{|\mc{S}_{x_0}|}{|\mc{S}_{y_0}|} \leq \left(1+4\cdot 2^{-L/(k\Delta^{2})}\right)r_+.\]
\end{proposition}
\begin{proof}
By iterating the condition (LP3), we have
\begin{align*}
\sum_{(x,y)\in \mc{L}_L: x \to \sigma} \h p_{x,y}^{x} &= 1 \quad \text{ for all }\sigma \in \mc{S}_{x_0},\\
\sum_{(x,y)\in \mc{L}_L: y \to \sigma} \h p_{x,y}^{y} &= 1 \quad \text{ for all }\sigma \in \mc{S}_{y_0}.
\end{align*}
Therefore, 
\begin{align*}
|\mc{S}_{x_0}| &= \sum_{\sigma\in \mc{S}_{x_0}} \sum_{(x,y)\in \mc{L}_L: x \to \sigma} \h p_{x,y}^{x},\\
|\mc{S}_{y_0}| &= \sum_{\sigma\in \mc{S}_{y_0}}\sum_{(x,y)\in \mc{L}_L: y \to \sigma} \h p_{x,y}^{y}.
\end{align*}
At the end of this subsection, we will prove the following. 
\begin{lemma}\label{lem:LP-error}
For all $p'' \le (100q^{2}k\Delta^{4})^{-1}$, $p' \leq p''/(100\Delta^{3}q)$, and $L \geq 8k\Delta^{2}$, 
\begin{align*}
\frac{1}{|\mc{S}_{x_0}|}\sum_{\sigma\in \mc{S}_{x_0}}\sum_{(x,y)\in \mc{L}^{b}_{L}:x \to \sigma} \h p_{x,y}^{x} \le 2^{-L/(k\Delta^2)},\\
\frac{1}{|\mc{S}_{y_0}|}\sum_{\sigma\in \mc{S}_{y_0}}\sum_{(x,y)\in \mc{L}^{b}_{L}:y \to \sigma} \h p_{x,y}^{y} \le 2^{-L/(k\Delta^2)}.
\end{align*}
\end{lemma}

Given this lemma, we have
\begin{align*}
    |\mc{S}_{x_0}|
    &= \sum_{\sigma\in \mc{S}_{x_0}} \sum_{(x,y)\in \mc{L}^{g}_L: x \to \sigma} \h p_{x,y}^{x} + \sum_{\sigma\in \mc{S}(x_0)} \sum_{(x,y)\in \mc{L}_L^{b}: x \to \sigma} \h p_{x,y}^{x}\\
    &= \left(\sum_{(x,y) \in \mc{L}_L^{g}}\h p^{x}_{x,y}\cdot |\mc{S}_{x}|\right) \pm |\mc{S}_{x_0}|\cdot 2^{-L/(k\Delta^{2})},
\end{align*}
where the first term follows by interchanging sums and the second term follows by \cref{lem:LP-error}. A similar estimate also holds for $|\mc{S}_{y_0}|$. 

Thus, we have
\begin{align*}
    \frac{|\mc{S}_{x_0}|\cdot(1\pm 2^{-L/(k\Delta^2)})}{|\mc{S}_{y_0}|\cdot(1\pm 2^{-L/(k\Delta^2)})}
    &= \frac{\sum_{(x,y) \in \mc{L}_L^{g}}\h p^{x}_{x,y}\cdot |\mc{S}_{x}|}{\sum_{(x,y) \in \mc{L}_L^{g}}\h p^{y}_{x,y}\cdot |\mc{S}_{y}|}\\
    & \in \left[\frac{r_-\cdot \sum_{(x,y) \in \mc{L}_L^{g}}\h p^{y}_{x,y}\cdot |\mc{S}_{y}|}{\sum_{(x,y) \in \mc{L}_L^{g}}\h p^{y}_{x,y}\cdot |\mc{S}_{y}|}, \frac{r_+\cdot \sum_{(x,y) \in \mc{L}_L^{g}}\h p^{y}_{x,y}\cdot |\mc{S}_{y}|}{\sum_{(x,y) \in \mc{L}_L^{g}}\h p^{y}_{x,y}\cdot |\mc{S}_{y}|}\right]\\
    & \in [r_-, r_+],
\end{align*}
where the second line follows from (LP2). Thus
\begin{align*}
    \frac{|\mc{S}_{x_0}|}{|\mc{S}_{y_0}|} \in [(1-4\cdot 2^{-L/(k\Delta^2)})r_-, (1+4\cdot 2^{-L/(k\Delta^2)})r_+],
\end{align*}
as desired. 
\end{proof}

\begin{proof}[Proof of \cref{lem:LP-error}]
We will only prove the statement for $|\mc{S}_{x_0}|$; the proof of the other statement is identical. 

For a node $(x,y) \in \mc{T}$, let $(V_S)'_{x,y} = (V_S)_{x,y}\setminus \{v_1^*,\dots, v_{\ell-1}^*\}$. For a node $(x,y) \in \mc{T}$ which is not a leaf, we will use $v_{x,y}$ to denote the variable such that the children of $(x,y)$ are defined on the set $(V_S)_{x,y} \cup \{v_{x,y}\}$. When $(x,y) \in \mc{T}$ is clear from context, we will denote $v_{x,y}$ simply by $v$.   

Consider the following way of generating random root-to-leaf paths of $\mc{T}_{L}$. At a non-leaf node $(x,y) \in \mc{T}$, sample a value for $v = v_{x,y}$ according to $\mu_S[v = \cdot \mid x]$ to generate an assignment $x'$ on $(V_S)_{x,y} \cup \{v_{x,y}\}$. Then, choose a random element $b'$ of $[q]$ and go to the node $(x',y_{v(b')}) \in \mc{T}$, where the probability of choosing each $b \in [q]$ is
\[p(x,y,x',y_{v(b)}) = \frac{\h p^{x'}_{x',y_{v(b)}}}{\h p^{x}_{x,y}}.\]
Note that by (LP3), $p(x,y,x',y_{v(\cdot)})$ is indeed a probability distribution. Let $(X,Y)$ denote the random leaf of $\mc{T}_L$ returned by this process and let $\h \mu$ denote the probability distribution on $\mc{L}_L$ induced by this process. 

Let $(x_f, y_f) \in \mc{L}_L$ and denote the corresponding root-to-leaf path by $(x_0, y_0), \dots, (x_f, y_f)$. Then,  
\[\h \mu[(X,Y) = (x_f, y_f)] = \prod_{t=1}^{f}\mu_S[x_{t} \mid x_{t-1}]\times \prod_{t=1}^{f}p(x_{t-1}, y_{t-1}, x_{t}, y_{t}) = \frac{|\mc{S}_{x_f}|}{|\mc{S}_{x_0}|} \cdot \frac{\h p^{x_f}_{x_f, y_f}}{\h p^{x_0}_{x_0, y_0}} = \frac{|\mc{S}_{x_f}|}{|\mc{S}_{x_0}|} \cdot \h p^{x_f}_{x_f, y_f},\]
where the final equality follows by (LP3).  

Therefore,
\begin{align*}
    \frac{1}{|\mc{S}_{x_0}|}\sum_{\sigma \in \mc{S}_{x_0}}\sum_{(x,y) \in \mc{L}_L^{b}:x\to \sigma}\h p^{x}_{x,y}
    &= \sum_{\sigma \in \mc{S}_{x_0}}\sum_{(x,y) \in \mc{L}_L^{b}:x\to \sigma}\frac{\h \mu[(X,Y) = (x, y)]}{|\mc{S}_{x}|}\\
    &= \sum_{(x,y) \in \mc{L}_L^{b}}\h \mu[(X,Y) = (x,y)]\\
    &\leq \h \mu[(X,Y) \in \{(x,y) \in \mc{T}_L: |(V_S)'_{x,y}| \geq L\}].
\end{align*}

In order to bound the quantity on the right, we will first find a more convenient combinatorial characterization of the event. For this, fix $(x,y) \in \mc{T}_L$ such that $|(V_S)'_{x,y}|\geq L$. Recall the notation $D_{x,y}, \mc{F}_{x,y}$ from (T2). 
We also need some further notation.
\begin{itemize}
    \item For $i \in [m]$, we say that $C_i$ is $x$-frozen if $\mb{P}[C_i \mid x] > p''$. We denote the set of $x$-frozen constraints by $\mc{F}'_{x}$.
    \item For $i \in [m]$, we say that $C_i$ has a disagreement if $\on{vbl}(C_i) \cap D_{x,y} \neq \emptyset$. Denote all such constraints by $\mc{D}_{x,y}$.
    \item For $i \in [m]$, we say that $C_i \in \mc{C}$ is bad if $\mc{C}_i \in \mc{F}_{x,y} \cup \mc{D}_{x,y}$. 
    We denote the set of bad constraints by $\mc{B}_{x,y}$. 
    \item $G(\mc{C}) = (\mc{C}, E)$ is the graph whose vertices are constraints in $\mc{C}$, and for $i\neq j$, there is an edge between $C_i$ and $C_j$ if and only if $\on{vbl}(C_i) \cap \on{vbl}(C_j) \neq \emptyset$.
    \item $G^2(\mc{C}) = (\mc{C}, E')$ is the graph whose vertices are constraints in $\mc{C}$, and for $i\neq j$, there is an edge between $C_i$ and $C_j$ if and only if there exists $k$ with $\on{vbl}(C_i)\cap \on{vbl}(C_k) \neq \emptyset$, and $\on{vbl}(C_j)\cap \on{vbl}(C_k) \neq \emptyset$.
\end{itemize}

\begin{claim}
$|\mc{B}_{x,y}| \geq L/(k\Delta)$.
\end{claim}
\begin{proof}
First, note that by (O3), (T3), and (T4), for every $v \in (V'_S)_{x,y}$, there exists some $C \in \mc{C}$ such that $\on{vbl}(C)\cap (V_D)_{x,y} \neq \emptyset$ and $v\in \on{vbl}(C)$. 

Next, note that by (T2), for every $v \in (V_D)_{x,y}$, there exists some $B \in \mc{B}_{x,y}$ such that $v \in \on{vbl}(B)$. 

Therefore,
\begin{align*}
    L 
    & \leq |(V_S')_{x,y}|\\
    & \leq k\cdot |\{C\in \mc{C} : \on{vbl}(C)\cap (V_D)_{x,y} \neq \emptyset\}|\\
    &\leq k\cdot |\{C\in \mc{C} : \on{vbl}(C)\cap \on{vbl}(B) \neq \emptyset\text{ for some }B \in \mc{B}_{x,y}\}|\\
    &\leq k\Delta \cdot |\mc{B}_{x,y}|. \qedhere
\end{align*}
\end{proof}

Fix $C^* \in \mc{C}$ such that $v_\ell^* \in \on{vbl}(C^*)$. 

\begin{claim}
There exists a $\{2,3\}$-tree $T \subseteq \mc{B}_{x,y}$ in $G(\mc{C})$ with $|T| \geq L/(k\Delta^2)$ and such that $T$ contains $C^*$.
\end{claim}
\begin{proof}
By (T4), (O3), and induction, the induced subgraph of $G^2(\mc{C})$ on  $\mc{B}_{x,y}$ is connected. 
Moreover, $C^* \in \mc{D}_{x,y} \subseteq \mc{B}_{x,y}$. Given this, the claim follows from \cref{lem:large-2-3-tree} and the previous claim.
\end{proof}

Let $\h L = L/(k\Delta^{2})$ and let $\mf{T}_{\h L}$ denote the collection of $\{2,3\}$-trees in $G(\mc{C})$ of size $\h L$ which contain $C^*$. Then, by the previous discussion,
\begin{align*}
    \h \mu[(X,Y) \in \{(x,y) \in \mc{T}_L: |(V_S)'_{x,y}| \geq L\}]
    &\leq \sum_{T \in \mf{T}_{\h L}}\h \mu[T \subseteq \mc{B}_{X,Y}]\\
    &\leq (e\Delta^{3})^{\h L}\cdot \h \mu[T \subseteq \mc{B}_{X,Y}],
\end{align*}
where the final inequality uses \cref{lem:number-2-3-trees}. 

Finally, let us fix $T \in \mf{T}_{\h L}$ and estimate
\[\h \mu[T \subseteq \mc{B}_{X,Y}].\]

We denote the vertices of $T$ in $G(\mc{C})$ by $C'_1 = C^*,C_2',\dots, C'_{\h L}$. Note, in particular, that $v_\ell^* \notin \on{vbl}(C'_j)$ for $2\leq j \leq \h L$. By multiplying the result by an overall factor of $2^{\h L}$, it suffices to bound the probability
\[\h \mu[C'_{2},\dots, C'_{t} \in \mc{D}_{X,Y} \wedge C'_{t+1},\dots, C'_{\h L} \in \mc{F}_{X,Y} \setminus \mc{D}_{X,Y}],\]
which is at most
\[\h \mu[C'_{2},\dots, C'_{t} \in \mc{D}_{X,Y} \wedge C'_{t+1},\dots, C'_{\h L} \in \mc{F}'_{X}].\]
Using the law of total probability, this is equal to
\begin{align}
\label{eqn:break}
\mb{E}_{\h \mu}\left[\h \mu [C'_{2},\dots, C'_{t} \in \mc{D}_{X,Y}\mid X]\cdot \h \mu [C'_{t+1},\dots, C'_{\h L} \in \mc{F}'_X \mid X]\right].
\end{align}
\begin{claim}
For any possible realization $x$ of $X$,
\[\h \mu [C'_2,\dots, C'_t \in \mc{D}_{x,Y}\mid x] \leq (4kq\eta)^{t-1}.\]
\end{claim}
\begin{proof}
Let $j \in [t]$. Given $x$, we have $C'_j \in \mc{D}_{x,Y}$ only if there is some $v \in \on{vbl}(C'_j)$ for which $Y(v) \neq x(v)$. Since $|\on{vbl}(C'_j)| \leq k$, we see by (LP4) that this happens with probability at most $4kq\eta$. Moreover, since $\on{vbl}(C'_j)$ are disjoint for different values of $[j]$, these events are independent, which gives the desired conclusion.   
\end{proof}
Using this claim, we see that the quantity in \cref{eqn:break} is bounded above by
\begin{align}
\label{eqn:x}
(4kq\eta)^{t-1}\cdot \h \mu [C'_{t+1},\dots, C'_{\h L} \in \mc{F}'_X].
\end{align}
\begin{claim}\label{claim:heavy-bound}
$\h \mu [C'_{t+1},\dots, C'_{\h L} \in \mc{F}'_X] \leq (2p'q/p'')^{\h L - t}.$
\end{claim}
\begin{proof}
Let $$\mf{X} = \{x: (x,y) \in \mf{T}_{\h L} \text{ for some } y \wedge C_{t+1}',\dots, C'_{\h L} \in \mc{F}'_x \},$$
so that our goal is to bound $\h \mu [\mf{X}]$. We will use an argument similar to \cref{lem:martingale}.

For $h \in \{0,\dots, L\}$, let $\mc{T}_{h} \subseteq \mc{T}_L$ denote those nodes $(x,y) \in \mc{T}$ for which $|(V_S)_{x,y}| \leq h + \ell$ and let $\mc{L}_h$ denote the leaves of $\mc{T}_h$. To any $(x,y) \in \mc{T}_{L}$, we can naturally associate a root-to-leaf path $(x_0,y_0),\dots, (x_L,y_L) = (x,y)$, where $(x_i,y_i) \in \mc{T}_i$ and nodes are possibly repeated at the end of the path. Interpreting $\h \mu$ as a distribution on root-to-leaf paths of this form, let $\mf{S}(h)$ denote the sigma-algebra induced on $\mc{T}_h$ and let
\[M_h = \prod_{j = t+1}^{\h L}\mb{P}[C_j'\mid X_h]\cdot \exp(-\eta q W(X_h, Y_h)),\]
where
\[W(X_h, Y_h) = \#\left\{j \in [h-1]: v_{X_j, Y_j} \in \on{vbl}(C'_{t+1}) \cup \dots \cup \on{vbl}(C'_{\h L}) \right\}.\]
Then, $M_h$ is measurable with respect to $\mf{S}(h)$ and using the argument  in \cref{lem:martingale}, it is readily seen that
\[\mb{E}_{\h \mu} [M_{h+1} \mid \mf{S}(h)] \leq  M_h.\]
Indeed, we only need to check that given $(X_h, Y_h)$ and letting $C_T$ denote the unique constraint (if any) containing $v = v_{X_h,Y_h}$, we have 
\begin{align*}
    \mb{E}_{\h \mu}\bigg[\mb{P}[C_T \mid X_{h+1}] \mid  \mf{S}(h)\bigg] 
    &= \sum_{a\in [q]} \frac{\mb{P}[C_T \wedge X_h \wedge v=a]\cdot {\h \mu}[v = a \mid X_h]}{\mb{P}[X_h \wedge v=a]}\\
    &= \sum_{a \in [q]}\frac{\mb{P}[C_T\wedge X_h \wedge v=a]}{\mb{P}[X_h]}\cdot \frac{\h \mu [v = a \mid X_h]}{\mb{P}[v = a]}\\
    &\leq \sum_{a \in [q]}\frac{\mb{P}[C_T\wedge X_h \wedge v=a]}{\mb{P}[X_h]}\frac{q^{-1} + \eta}{q^{-1}} \\
    &\leq \exp(\eta q)\mb{P}[C_T \mid X_h],
\end{align*}
where the third line follows from (2) of \cref{lem:tree-properties}.
Since
\[\mb{P}[C_j'\mid X_0] \leq p'q,\]
it therefore follows that
\[\mb{E}_{\h \mu}[M_{L}] \leq (p'q)^{\h L - t}.\]
Also, for any $(x,y) \in \mc{T}_L$, we have that $W(x,y) \leq k(\h L - t)$. Hence, by Markov's inequality and the assumption on $p''$, 
\begin{align*}
    \h \mu[\mf{X}] 
    &\leq \h \mu[M_L \geq (p'')^{\h L - t}\exp(-\eta q k(\h L - t))]\\
    &\leq \left(\frac{p'q\exp(\eta q k)}{p''}\right)^{\h L - t} \leq \left(\frac{2p'q}{p''}\right)^{\h L - t}. \qedhere
\end{align*}
\end{proof}

Using the preceding claim along with \cref{eqn:x} and simplying using the bounds on $p', p''$ completes the proof. 
\end{proof}

\section{Proof of \cref{thm:counting,thm:sampling}}
\label{sec:proofs}
We are now ready to prove \cref{thm:counting,thm:sampling}. The algorithms in this section exploit a refinement of the analysis in the previous section, which we present in \cref{sec:mod-lp}. Following this, we present the proof of \cref{thm:counting} in \cref{sec:counting} and the proof of \cref{thm:sampling} in \cref{sec:sampling}. We will freely use the notation introduced in the previous two sections. 


\subsection{Refined analysis of the linear program}
\label{sec:mod-lp}
As in \cref{sec:algorithm}, fix an ordering $v_1,\dots, v_n$ of the variables. Let $p'\le p''$ be parameters to be chosen later. Recall that in \cref{subsec:Step-1}, we generate a partial assignment on a subset $v_1^*,\dots, v_s^*$ of the variables (where $s$ is itself random) and ``freeze'' the remaining variables. This process depends on the parameter $p'$ in (R4). As before, for $i \in [s]$, we let $P_i$ denote the partial assignment on $v_1^*,\dots, v_i^*$ and we let $\nu$ denote the distribution on partial assignments given by the final partial assignment $P_s$. Once again, we emphasize that $s$ is random. 

For our approximate sampling algorithm, we will also need the following variation of this procedure. We consider the same randomized greedy procedure as in \cref{subsec:Step-1}, except now, in (R3), we assign $v_i^*$ a random value chosen according to the distribution $\mu_S[v_i^* = \cdot \mid P_{i-1}]$.  For now, the reader should ignore the question of how to efficiently implement such a procedure. We let $\nu_S$ denote the distribution on partial assignments given by the final partial assignment $P_s$, noting again that $s$ is random. 

Given $v_1^*,\dots, v_i^*$ and $a\in [q]$, we denote by $P_i(a)$ the partial assignment extending $P_{i-1}$ by setting $v_i^*=a$. As before, for $h\in [n]$, we let $\iota(h)$ be the largest index $i$ such that $v_i^*=v_w$ for some $w\le h$ i.e.~$\iota(h)$ is the number of variables among $\{v_1,\dots, v_h\}$ assigned values by the partial assignment. 

For each variable $v$, there are at most $\Delta $ constraints $C \in \mc{C}$ such that $v \in \on{vbl}(C)$. 
Let $\h L = L/(k\Delta^2)$, and let $\mf{T}_{\h L,v}$ be the set of $\{2,3\}$-trees in $G(\mc{C})$ of size $\h L$ containing one of these constraints. For any $T \in \mf{T}_{\h L, v}$, we let $C^*_v$ denote the unique $C \in T$ satisfying $v \in \on{vbl}(C)$. Recall that $\mc{F}_{x,y}$ denotes the constraints $C \in \mc{C}$ for which $\mb{P}[C \mid x] > p''$ or $\mb{P}[C \mid y] > p''$. 



For $\ell \in [s]$, we define the idealized decision tree $\mc{T}$ starting from the root node $(x_0,y_0) := (P_\ell(a),P_\ell(b))$ and the $L$-truncated decision tree $\mc{T}_L$ as in \cref{sec:estimation-marginal}. Let $\eta=(1-3p''q)^{-\Delta} - 1$. The following lemma was proved during the course of the proof of \cref{lem:LP-error}.
\begin{lemma}\label{lem:LP-error-2}
For all $p'=p'' \le (1000q^{2}k\Delta^{4})^{-1}$ and $p \leq p''/(1000\Delta^{3}q)$,
\begin{align*}
\frac{1}{|\mc{S}_{x_0}|}\sum_{\sigma\in \mc{S}_{x_0}}\sum_{(x,y)\in \mc{L}^{b}_{L}:x \to \sigma} \h p_{x,y}^{x} \le \sum_{T \in \mf{T}_{\h L , v^*_\ell}} \,\, \prod_{C \in T, C\ne C^*_{v^*_\ell}}\left(4kq \eta+\frac{2\mb{P}[C\mid P_{\ell-1}]}{p''}\right),\\
\frac{1}{|\mc{S}_{y_0}|}\sum_{\sigma\in \mc{S}_{y_0}}\sum_{(x,y)\in \mc{L}^{b}_{L}:y \to \sigma} \h p_{x,y}^{y} \le \sum_{T \in \mf{T}_{\h L , v^*_\ell}} \,\, \prod_{C \in T, C\ne C^*_{v^*_\ell}}\left(4kq \eta+\frac{2\mb{P}[C\mid P_{\ell-1}]}{p''}\right).
\end{align*}
\end{lemma}
\begin{remark}
Note that there is no factor of $q$ multiplying $2\mb{P}[C \mid P_{\ell-1}]/p''$ since $v_\ell^* \notin \on{vbl}(C)$ for any $C$ that features in the product. 
\end{remark}

For $h \in [n+1]$, let $\mc{E}(h)$ denote the event that for all variables $v \in V$,
\begin{align*}
    \sum_{T \in \mf{T}_{\h L , v}} \,\, \prod_{C \in T, C\ne C^*_{v}}\left(4kq\eta+\frac{2\mb{P}[C \mid P_{\iota(h-1)}]}{p''}\right) \le n^42^{-L/(k\Delta^2)}.
\end{align*}
The next lemma, together with Markov's inequality, shows that $\mc{E}(h)$ occurs with high probability with respect to both $\nu$ and $\nu_S$.
\begin{lemma}\label{lem:LP-error-aggregate}
Let $p'=p''\le (100q^{2}k\Delta^{4})^{-1}$, $p \leq p''/(100\Delta^{3})$, and $L \geq 8k\Delta^{2}$. Then, for all $h \in [n+1]$ and $v \in V$,
\begin{align*}
\mb{E}_\nu\left[\sum_{T \in \mf{T}_{\h L , v}} \,\, \prod_{C \in T, C\ne C^*_{v}}\left(4kq\eta+\frac{2\mb{P}[C \mid P_{\iota(h-1)}]}{p''}\right)\right] \le 2^{-L/(k\Delta^2)}, 
\end{align*}
and
\begin{align*}
\mb{E}_{\nu_S}\left[\sum_{T \in \mf{T}_{\h L , v}} \,\, \prod_{C \in T, C\ne C^*_{v}}\left(4kq\eta+\frac{2\mb{P}[C \mid P_{\iota(h-1)}]}{p''}\right)\right] \le 2^{-L/(k\Delta^2)}.
\end{align*}
\end{lemma}
\begin{proof}
We will first prove the statement for $\nu$. Fix $v \in V$ and observe that for any $T \in \mf{T}_{\h L, v}$, the sets
$\on{vbl}(C)$ are disjoint for $C\in T$. Let $\mf{S}(t)$ denote the $\sigma$-algebra generated by the output of the randomized greedy procedure on partial assignments of $v_1,\dots,v_t$. Then, by an identical argument to \cref{lem:martingale}, we see that
\[
M_t = \prod_{C \in T, C\ne C^*_{v}}\left(4kq\eta+\frac{2\mb{P}[C\mid P_{\iota(t-1)}]}{p''}\right)
\]
satisfies
\[\mb{E}_{\nu}[M_{t + 1} \mid \mf{S}(t)] = M_{t}\]
Thus, for any $h \in [n+1]$,
\[
\mb{E}_\nu\left[\prod_{C \in T, C\ne C^*_{v}}\left(4kq\eta+\frac{2\mb{P}[C\mid P_{\iota(h-1)}]}{p''}\right)\right] \leq \left(4kq\eta+\frac{2p}{p''}\right)^{|T|-1},
\]
so that by linearity of expectation,
\begin{align*}
\mb{E}_\nu\left[\sum_{T \in \mf{T}_{\h L , v}} \,\, \prod_{C \in T, C\ne C^*_{v}}\left(4kq\eta+\frac{2\mb{P}[C\mid P_{\iota(h-1)}]}{p''}\right)\right] 
&\leq |\mf{T}_{\h L,v}| \left(4kq\eta+\frac{2p}{p''}\right)^{|T|-1}\\
&\le \Delta\cdot(e\Delta^3)^{\h L}\cdot \left(4kq\eta+\frac{2p}{p''}\right)^{\h L-1}.
\end{align*}
The desired bound is obtained by using the assumptions on $p$ and $p''$. 

Next, we prove the statement for $\nu_S$. Fix $v \in V$. For $T \in \mf{T}_{\h L, v}$, let $W_t(T)$ be the number of variables among $v_1^*,\dots,v_{\iota(t)}^*$ that are contained in some $C\in T$. Also, as before, let $\mf{S}(t)$ denote the $\sigma$-algebra generated by the output of the randomized greedy procedure on partial assignments of $v_1,\dots, v_t$. Then, as in \cref{claim:heavy-bound}, we have that 
\[
M_t = \prod_{C \in T, C\ne C^*_{v}}\left(4kq\eta+\frac{2\mb{P}[C\mid P_{\iota(t-1)}]}{p''}\right) \cdot \exp(-\eta q W_{t-1}(T))
\]
satisfies
\[\mb{E}_{\nu_S}[M_{t + 1} \mid \mf{S}(t)] \le  M_{t}.\]
Thus, for any $h \in [n+1]$, 
\[
\mb{E}_{\nu_S}\left[\prod_{C \in T, C\ne C^*_{v}}\left(4kq\eta+\frac{2\mb{P}[C \mid P_{\iota(h-1)}]}{p''}\right)\right] \leq \left(4kq\eta+\frac{2p}{p''}\right)^{|T|-1}\exp(\eta q k)^{|T|-1}.
\]
so that by linearity of expectation,
\begin{align*}
\mb{E}_{\nu_S}\left[\sum_{T \in \mf{T}_{\h L , v}} \,\, \prod_{C \in T, C\ne C^*_{v}}\left(4kq\eta+\frac{2\mb{P}[C \mid P_{\iota(h-1)}]}{p''}\right)\right] 
&\leq |\mf{T}_{\h L,v}| \left(8kq\eta+\frac{4p}{p''}\right)^{|T|-1}\\
&\le \Delta\cdot (e\Delta^3)^{\h L}\cdot \left(8kq\eta+\frac{4p}{p''}\right)^{\h L-1}.
\end{align*}
The desired bound is obtained by using the assumptions on $p$ and $p''$. 
\end{proof}

Combining \cref{lem:LP-error-2,lem:LP-error-aggregate} and the analysis in \cref{sec:estimation-marginal}, we obtain the following proposition.
\begin{proposition}\label{prop:lp-certify-2}
Let $p'=p'' \le (1000q^{2}k\Delta^{4})^{-1}$, $p \leq p''/(1000\Delta^{3})$, and $\delta \in (0,1)$. Then, for $L\ge 8k\Delta^2 \log (n/\delta)$, the event $\wedge_{h \in [n+1]}\mc{E}(h)$ has probability at least $1-(\delta/n^{2})$ with respect to both $\nu$ and $\nu_S$.  

Moreover, on the event $\wedge_{h \in [n+1]}\mc{E}(h)$, 
the feasibility of the LP with parameters $r_-\leq r_+$ implies that 
\[\left(1-4\cdot n^42^{-L/(k\Delta^{2})}\right)r_- \leq\frac{|\mc{S}_{x_0}|}{|\mc{S}_{y_0}|} \leq \left(1+4\cdot n^42^{-L/(k\Delta^{2})}\right)r_+.\]
\end{proposition}

\subsection{Approximate counting: proof of \cref{thm:counting}} 
\label{sec:counting}
We have the following analogue of \cref{lem:derandomize}. 
\begin{lemma}
\label{lem:derandomize-2}
Let $p'=p'' \le (1000q^{2}k\Delta^{4})^{-1}$, $p \leq p''/(1000\Delta^{3})$, and $L = 80k\Delta^2 \log (\Delta n)$. There exists a deterministic algorithm running in time $O(n^{\on{poly}(\log q, \Delta, k)})$ which generates a sequence of partial assignments $P_1,\dots, P_s$ with the following properties. 
\begin{enumerate}
    \item For all $i \in [s]$, $P_i$ assigns values to $i$ variables, and $P_i$ extends $P_{i-1}$. 
    \item $A_{s} = \emptyset$.
    \item For all $i \in [s]$ and $j \in [m]$, $\mb{P}[C_j \mid P_i] \leq p'q$. 
    \item Every connected component in $G(\mc{F})$ has size at most  $ L/(k\Delta)$.
    \item The event $\wedge_{h \in [n+1]}\mc{E}(h)$ is satisfied.  
\end{enumerate}
\end{lemma}

\begin{proof}
The proof is a modification of the proof of \cref{lem:derandomize}. Let $L' = L/(k\Delta^2)$. Recall that for each variable $v$, $\mf{T}_{L',v}$ is the set of $\{2,3\}$-trees in $G(\mc{C})$ of size $L'$ containing $C^*_v$. Let $\mf{T}$ denote the collection of all \{2,3\}-trees of size $L'$ in $G(\mc{C})$. Note that $|\mf{T}_{L',v}|\le |\mf{T}| \leq \on{poly}(n^{\log^{2}\Delta})$ and the collections $\mf{T}_{L',v}$, $\mf{T}$ can be constructed in time $\on{poly}(n^{\log^{2}\Delta})$.

For a partial assignment $X$, define
\[H(X) = \sum_{v \in V} \sum_{T \in \mf{T}_{L' , v}} \,\, \prod_{C \in T, C\ne C^*_{v}}\left(4kq\eta+\frac{2\mb{P}[C\mid X]}{p''}\right) .\]
 Note that 
\[
H(X) \ge \sum_{T \in \mf{T}}\prod_{C \in T} \left(\frac{\mb{P}[C \mid X]}{p''}\right).
 \]
As in the proof of \cref{lem:derandomize}, if we can find a sequence of partial assignments $P_1,\dots, P_s$ satisfying properties (1), (2), (3) such that $H(P_1),\dots, H(P_s) < n^42^{-L/(k\Delta^2)} < 1$, then (4) and (5) are also satisfied.

For this, we follow the same greedy procedure as in \cref{subsec:Step-1}, except now, after having chosen $P_{i-1}$ and $v_i^*$, we choose the value of $v_i^*$ in (R3) to be
\[\argmin_{a \in [q]}H(P_{i-1} \wedge {v_i^* = a}).\]
Similar to \cref{lem:derandomize}, this ensures that $H(P_{i}) \leq H(P_{i-1})$ for all $i \in [s]$. Thus, it is possible to choose $P_i$ to ensure $H(P_{i-1}) \leq H(P_i)$. Finally, since
$$H(\emptyset) \leq n\cdot \Delta\cdot (e\Delta^{3})^{L'}\cdot \left(4kq\eta+\frac{2p}{p''}\right)^{L'} < n^42^{-L/(k\Delta^2)},$$
we are done. 
\end{proof}

Finally, given partial assignments $P_1,\dots, P_s$ satisfying the properties of \cref{lem:derandomize-2}, we can use \cref{lem:initial-factorization}, \cref{prop:lp-certify-2} and the analysis of \cref{sec:estimation-marginal} to complete the proof of \cref{thm:counting}. 

\subsection{Approximate sampling: proof of \cref{thm:sampling}}
\label{sec:sampling}
We consider the following sampling procedure. Fix a parameter $\varepsilon \in (0,1)$. Let $L = 80k\Delta^{2}\log(\Delta n /\varepsilon)$. 
\begin{enumerate}[(S1)]
    \item Fix an arbitrary ordering $v_1,\dots, v_n$ of the variables. Initialize the set of frozen variables $F_0 = \emptyset$, the set of available variables $A_0 = V$, $\iota(0) = 0$, and $P_{0} = \emptyset$.
    \item Let $1\leq i \leq n$. Given $P_{\iota(i-1)}, F_{i-1}, A_{i-1}$, if $\mc{E}(i)$ does not hold, then output an arbitrary satisfying assignment (which can be found using the algorithmic LLL in \cite{moser2010constructive}) and terminate. Otherwise, $\mc{E}(i)$ holds.
    \item If $v_i \notin A_{i-1}$, then $\iota(i) = \iota(i-1), F_{i} = F_{i-1}, A_{i} = A_{i-1}$. Increment $i$. If $i\le n$, return to (S2). Otherwise, proceed to (S6).
    \item If $v_i \in A_{i-1}$, then approximate the marginal $\mu_{S}[v_i = \cdot \mid P_{\iota(i-1)}]$ within total variation distance $\varepsilon/(8n)$ using the LP. Then, assign $v_i$ a random value in $[q]$ distributed according to the output of the LP. Let $\iota(i) = \iota(i-1) + 1$, and $P_{\iota(i)}$ be the extension of $P_{\iota(i-1)}$ resulting from assigning a value to $v_i$.
    \item Let 
    \[\mc{F}_i = \{j\in [m]: \mb{P}[C_j \mid P_{\iota(i)}] > p''\}.\]
    Set
    \[F_{i} = F_{i-1} \cup \bigcup_{j \in \mc{F}_i}(\on{vbl}(C_j) \cap A_{i}\setminus \{v_i\})\text{ and }A_i = A_{i-1} \setminus (F_{i} \cup \{v_i\}).\]
    Increment $i$. If $i \leq n$, return to (S2). Otherwise, proceed to (S6).
    \item Let $\mc{F} = \cup_{i\in [n]}\mc{F}_i$. Consider $G^{2}(\mc{F})$, which is the induced subgraph of $G^{2}(\mc{C})$ by the vertices $\mc{F}$. If any connected component of $G^{2}(\mc{F})$ has size larger than $80\Delta\log(\Delta n)$, then output an arbitrary satisfying assignment and terminate. 
    \item Else, use exhaustive enumeration to uniformly sample a satisfying assignment of unassigned variables appearing in each separate connected component of $G^{2}(\mc{F})$, and return the complete satisfying assignment thus obtained.
\end{enumerate}

It is immediate that the running time of the algorithm is as claimed in \cref{thm:sampling}. Let $\mu_{\on{Alg}}$ denote the distribution on satisfying assignments generated by the algorithm. We show that $\on{TV}(\mu_{\on{Alg}}, \mu_{S}) \leq \varepsilon$. 

For this, we begin by observing that if we could sample from the true marginal distribution $\mu_S[v_i = \cdot \mid P_{\iota(i-1)}]$ in (S4) and if we could output a uniform satisfying assignment extending the current partial assignment for the early termination in (S2) and (S6), then the resulting distribution on satisfying assignments output by the algorithm clearly coincides with $\mu_S$. 

Next, since the approximate marginals in (S4) are within $\varepsilon/(8n)$ of the true marginals, it follows from \cref{prop:lp-certify-2} and \cref{lem:small-comp} that the early termination condition in (S2) and (S6) occurs with probability at most $\varepsilon/4$. Therefore, $\on{TV}(\mu_{\on{Alg}}, \mu_{\on{Alg}'}) \leq \varepsilon/4$, where $\on{Alg}'$ denotes the sampling algorithm which is the same as above, except upon early termination in (S2) and (S6), we output a uniformly random satisfying assignment extending the current partial assignment. 

 Finally, since the approximate marginals in (S4) are within $\varepsilon/(8n)$ of the true marginals, it follows that $\on{TV}(\mu_{\on{Alg}'}, \mu_S) \leq \varepsilon/8$, so that by the triangle inequality, $\on{TV}(\mu_{\on{Alg}}, \mu_S) \leq \varepsilon$, as desired. 






\bibliographystyle{alpha}
\bibliography{main.bib}

\end{document}